\documentclass[10pt]{article}
\usepackage{amsfonts}
\usepackage{amssymb,amsmath,amsthm,latexsym}
\usepackage{amsmath,amsthm,amssymb}
\newtheorem{Theorem}{Theorem}[section]
\newtheorem{lem}[Theorem]{Lemma}
\newtheorem{Remark}[Theorem]{Remark}

\newtheorem{Corollary}[Theorem]{Corollary}
\newtheorem{Proposition}[Theorem]{Proposition}
\newtheorem{Example}[Theorem]{Example}
\numberwithin{equation}{section}

\begin{document}
\title{Three new classes of optimal  frequency-hopping sequence sets\footnote{
 E-Mail addresses: bocong\_chen@yahoo.com (B. Chen), L\_R\_Lin86@163.com (L. Lin), lingsan@ntu.edu.sg  (S. Ling), hwliu@mail.ccnu.edu.cn (H. Liu).}}

\author{Bocong Chen$^{1,2}$, ~Liren Lin$^3$, ~San Ling$^2$, ~Hongwei Liu$^3$}

\date{\small
${}^1$School of Mathematics, South China University of Technology, Guangzhou, Guangdong, 510641, China\\
${}^2$Division of Mathematical Sciences, School of Physical \&
Mathematical Sciences,
         Nanyang Technological University, Singapore 637616, Singapore\\
${}^3$School of Mathematics and Statistics,
Central China Normal University,
Wuhan, Hubei, 430079, China\\         }

\maketitle

\begin{abstract}

The study of frequency-hopping sequences (FHSs) has been focused on the establishment of
theoretical bounds for the parameters of FHSs as well as on  the construction of
optimal FHSs with respect to the bounds. Peng and Fan (2004) derived two lower bounds on the maximum nontrivial Hamming correlation
of an FHS set, which is    an important indicator in measuring the performance of an FHS set employed in practice.

In this paper, we obtain two main results. We study the construction of new optimal frequency-hopping sequence sets by using cyclic codes over finite fields.
Let $\mathcal{C}$ be a cyclic code of length $n$ over a finite field $\mathbb{F}_q$ such that $\mathcal{C}$
contains the one-dimensional   subcode
$
\mathcal{C}_0=\{(\alpha,\alpha,\cdots,\alpha)\in \mathbb{F}_q^n\,|\,\alpha\in \mathbb{F}_q\}.
$
Two codewords of $\mathcal{C}$ are said to be equivalent if
one can be obtained from the other through applying the cyclic shift a certain number of times.
We present a necessary and sufficient condition under which  the equivalence class of any codeword  in
$\mathcal{C}\setminus\mathcal{C}_0$ has size $n$.  This result   addresses an open question
raised by Ding {\it et al.} in \cite{Ding09}.
As a consequence, three new classes of optimal   FHS  sets with respect to
the Singleton bound are obtained,
some of which  are also optimal with respect to the Peng-Fan bound at the same time. We also show that the two Peng-Fan bounds
are, in fact, identical.


\end{abstract}

\section{Introduction}

Let $\ell$ be a positive integer and let $\mathbb{F}=\{f_0, f_1, \cdots, f_{\ell-1}\}$ be an alphabet of $\ell$ available frequencies.
A sequence $X=\{x_t\}_{t=0}^{n-1}$ is called  a {\it frequency-hopping sequence} (FHS) of length $n$ over $\mathbb{F}$ if $x_t\in \mathbb{F}$ for all $0\leq t\leq n-1$.
For
two FHSs  $X=\{x_t\}_{t=0}^{n-1}$ and $Y=\{y_t\}_{t=0}^{n-1}$ of length $n$ over $\mathbb{F}$,
if $x_t=y_t$ for all $0\leq t\leq n-1$,  then we say $X=Y$.
Let $\mathcal{S}$ be the set
of all FHSs of length $n$ over $\mathbb{F}$.
Any subset of $\mathcal{S}$ is called an {\it FHS set}. For any $X,Y\in \mathcal{S}$,
their {\it Hamming
correlation} is defined by
$$
H_{X,Y}(t)=\sum\limits_{i=0}^{n-1}h[x_i,y_{i+t}],~~ 0\leq t\leq n-1
$$
where $h[a,b]=1$ if $a=b$ and $0$ otherwise, and the subscript addition is taken modulo $n$.
 For any  distinct
$X, Y\in \mathcal{S}$, we have the following    measures,
$$
H(X)=\max\limits_{1\leq t<n}\big\{H_{X,X}(t)\big\}
$$
and
$$
H(X,Y)=\max\limits_{0\leq t<n}\big\{H_{X,Y}(t)\big\}.
$$
Let $\mathcal{F}$ be an FHS set  containing $N$ elements. The {\it maximum
nontrivial Hamming correlation} of the FHS  set  $\mathcal{F}$ is defined by
$$
M(\mathcal{F})=\max\Big\{\max\limits_{X\in \mathcal{F}}H(X), \max\limits_{X,Y\in \mathcal{F},X\neq Y}H(X,Y)\Big\}.
$$
As in \cite{Ding09}, we use $(n,N,\lambda;\ell)$   to denote the FHS set $\mathcal{F}$ with $N$ elements  of length $n$ over an alphabet of size $\ell$, where $\lambda=M(\mathcal{F})$.

Frequency-hopping spread spectrum  techniques
have been widely used  in modern communication systems, such as ultrawideband communications,
military communications,
Wi-Fi and  Bluetooth. In these systems, the receivers  may be  confronted with the interference
caused by undesired signals.
It is often
desirable   to properly select   frequency-hopping sequences (FHSs)   or   FHS  sets to mitigate the interference.
In order to evaluate the goodness of  an FHS set, the maximum
nontrivial Hamming
correlation   is   an important indicator.
In general, it is desirable  to construct an    FHS set   with a large
set size and a low Hamming correlation value, when its length and the
number of available frequencies are fixed. However, the parameters of an   FHS set are not  independent of one another, and they
are subjected to certain theoretical limits (e.g.,  see  \cite{Ding09}, \cite{Han}, \cite{Lempel}, \cite{Peng}, \cite{Yang}, \cite{Zhang}).
Lempel and Greenberger obtained a lower bound on $H(X)$ of any individual FHS $X$ in 1974 \cite{Lempel}.
Extending the Lempel-Greenberger bound,  Peng and Fan derived the following lower bounds on the maximum
nontrivial Hamming correlation of an
FHS set (\cite[Corollary 1]{Peng}).

For a real number $a$, let
$\lceil a\rceil$   denote the least integer not less
than $a$ and let
$\lfloor a\rfloor$ denote the integer part of $a$.
\begin{lem} (Peng-Fan bounds, \cite{Peng})
Let   $\mathcal{F}\subseteq \mathcal{S}$  be a set of $N$ sequences of length $n$ over an alphabet of size $\ell$. Define
$I=\lfloor nN/\ell\rfloor$. Then
\begin{equation}\label{PF1}
M(\mathcal{F})\geq\left\lceil\frac{(nN-\ell)n}{(nN-1)\ell}\right\rceil
\end{equation}
and
\begin{equation}\label{PF}
M(\mathcal{F})\geq\left\lceil\frac{2InN-(I+1)I\ell}{(nN-1)N}\right\rceil.
\end{equation}
\end{lem}

\begin{Remark}{\rm
Yang {\it et al.}  compared the above two Peng-Fan bounds in \cite{Yang}; it was shown   that the
Peng-Fan bound of  (\ref{PF}) may be tighter than that of (\ref{PF1}). However, the authors failed to find   examples
where   the bound of  (\ref{PF}) is strictly  tighter than that of  (\ref{PF1}), and finally suggested  that the exact relationship between the
bounds (\ref{PF1}) and (\ref{PF}) needs to be studied further.
In  this paper, we show that  the two Peng-Fan bounds
are, in fact, identical. (See Theorem \ref{p-f-iden} below. It is reasonable to assume that
$nN\geq\ell$ and its proof is deferred to the Appendix.)}
\end{Remark}
\begin{Theorem}\label{p-f-iden}
Let   $\mathcal{F}\subseteq \mathcal{S}$  be a set of $N$ sequences of length $n$ over an alphabet of size $\ell$. Define
$I=\lfloor nN/\ell\rfloor$. If $nN\geq\ell$, then
\begin{equation*}
M(\mathcal{F})\geq\left\lceil\frac{(nN-\ell)n}{(nN-1)\ell}\right\rceil=\left\lceil\frac{2InN-(I+1)I\ell}{(nN-1)N}\right\rceil.
\end{equation*}
\end{Theorem}

Besides the bounds on the Hamming correlation,
several  bounds on the size of an FHS set were also established.
Ding {\it et al.} in \cite{Ding09} obtained a number of bounds on the size of  an   FHS set
from certain classical bounds in  coding theory.

\begin{lem}
(Sphere-packing bound on the size of FHS sets, \cite{Ding09})
For any $(n, N, \lambda;\ell)$ FHS set $\mathcal{F}$,
where $\lambda<n$ and $\ell>1$, we have
\begin{equation}\label{sbound}
N\leq\left\lfloor\frac{\ell^n}{n\Big(\sum_{i=0}^{\left\lfloor (n-\lambda-1)/2\right\rfloor}{n \choose i}(\ell-1)^i\Big)}\right\rfloor.
\end{equation}
\end{lem}

\begin{lem}
(Singleton bound on the size of  FHS sets, \cite{Ding09})
For any $(n, N, \lambda;\ell)$ FHS set $\mathcal{F}$,
where $\lambda<n$ and $\ell>1$, we have
\begin{equation}\label{sinbound}
N\leq\left\lfloor\frac{\ell^{\lambda+1}}{n}\right\rfloor.
\end{equation}
\end{lem}

An  FHS set is called {\it optimal} if  one of the bounds  (\ref{PF1})-(\ref{sinbound}) is met.
It is of great interest to construct optimal FHS sets with respect to the bounds.
In recent years, numerous   constructions of
optimal      FHS sets     have been proposed
(e.g., see \cite{Cai}-\cite{Ge}, \cite{YKHan}, \cite{Liu}, \cite{Ren},  \cite{Yang2}-\cite{Zhou2}, and references therein).
Ding {\it et al.}  generalized the   ideas in \cite{Song} to construct optimal   FHS sets    by
using some special classes of cyclic  codes \cite{Ding09}. This idea was further investigated in \cite{Ding10} to obtain more optimal
FHS sets.    As shown in \cite{Ding09}, there is a natural equivalence relation   defined on any cyclic   code:
two codewords of a cyclic code are said to be equivalent if
 one can be obtained from the other through applying the cyclic shift a certain number of times.
A special class of cyclic codes $\mathcal{C}_{(q,m)}$ of length $n=(q^m-1)/(q-1)$ over a finite field $\mathbb{F}_q$   containing
$
\mathcal{C}_0=\{(\alpha,\alpha,\cdots,\alpha)\in \mathbb{F}_q^n\,|\,\alpha\in \mathbb{F}_q\}
$
was discussed in \cite{Ding09}. It was shown in \cite{Ding09} that
if the code length $n$ is a prime number, then  the equivalence class of any codeword  in
$\mathcal{C}_{(q,m)}\setminus\mathcal{C}_0$ has size $n$;
an  FHS set is thus obtained by taking exactly one element from  every equivalence class of $\mathcal{C}_{(q,m)}\setminus\mathcal{C}_0$,  which turns out  to be optimal with respect to  the sphere-packing bound (\ref{sbound}).
A natural  open question posed in \cite[p.3302]{Ding09} is whether the prime-length  constraint can be dropped   without changing the situation that
the equivalence class of any codeword  in $\mathcal{C}_{(q,m)}\setminus\mathcal{C}_0$ has size $n$.

In  this paper, we further explore the above idea   to   construct more optimal   FHS sets
by using  maximum distance separable (MDS)   cyclic codes.
Let $\mathbb{F}_q$ be the finite field with $q$ elements and let $n$ be a positive integer co-prime to $q$.
Assume that $\mathcal{C}$ is a
cyclic code of length $n$ over $\mathbb{F}_q$ containing $\mathcal{C}_0$,
where $\mathcal{C}_0=\{(\alpha,\alpha,\cdots,\alpha)\in \mathbb{F}_q^n\,|\,\alpha\in \mathbb{F}_q\}$.
In Section $3$, we present a necessary and sufficient condition
under which the equivalence class of any codeword  in
$\mathcal{C}\setminus\mathcal{C}_0$ has size $n$.
This result   addresses the aforementioned open question. Actually, it turns out that  the prime-length  constraint
is necessary and cannot be dropped (see Corollary \ref{open} in Section $3$).
In Section $4$, using the results  in Section $3$,  we obtain three   new classes of  optimal    FHS sets with respect to
the Singleton bound  (\ref{sinbound}),
some of which  are also optimal with respect to the Peng-Fan bound     (\ref{PF}) at the same time.
More precisely, the parameters of   the new   FHS sets are given as follows:
\begin{itemize}
\item[{(i)}]
$$\Big(q+1,\frac{q^{2k+1}-q}{q+1},2k;q\Big)$$
where $q=2^m$ with $m>1$ being a positive integer, and where
$1\leq k\leq\min\{p-1,2^{m-1}\}$ with $p$ being the smallest prime divisor of $q+1$.
This     FHS set is optimal with respect to the Singleton bound  (\ref{sinbound}).

\item[{(ii)}]
$$\Big(q+1,q(q-1),2;q\Big)$$
where $q$ is an odd  prime power.
The parameters of this   FHS set
meet both the Peng-Fan bound     (\ref{PF}) and the Singleton  bound  (\ref{sinbound}) at the same time.

\item[{(iii)}]
$$\Big(n,(q^{2k+2}-1)/n,2k+1;q\Big)$$
where
$n>1$ is an odd divisor of $q+1$   with $q$ being a prime power,
and where $k$ is an integer such that  $0\leq k\leq(n-3)/2-M$ with
$M\le (n-3)/2$ being the largest integer such that $\gcd(M,n)>1$. This    FHS set is optimal with respect to the Singleton  bound (\ref{sinbound}).
In particular, by taking $k=0$,
we have an  FHS set with parameters
$$\Big(n,(q^{2}-1)/n,1;q\Big)$$
which meet both the Peng-Fan bound     (\ref{PF}) and the Singleton  bound (\ref{sinbound}).
\end{itemize}

\section{Some Facts about Cyclic Codes }
In this section, we review some basic notation and results
about cyclic codes over finite fields. For the details, the reader is referred to  \cite{Ling} or \cite{Mabook}.

Let $\mathbb{F}_q$ be the finite field with $q$ elements and let $n$ be a positive integer co-prime to $q$.
A linear code $\mathcal{C}$ of length $n$ over $\mathbb{F}_q$ is called {\em cyclic}
if it is an ideal of $\mathbb{F}_q[x]/\langle x^n-1\rangle$.
It follows that any cyclic code $\mathcal{C}$ of length $n$ over $\mathbb{F}_q$ is generated uniquely
by a monic divisor $g(x)\in \mathbb{F}_q[x]$ of $x^n-1$, which is referred to as the
{\em generator polynomial},
and $h(x)=(x^n-1)/g(x)$ is called its {\em parity-check polynomial}.
We then know that the irreducible factors
of $x^n-1$ in $\mathbb{F}_q[x]$ determine  all  cyclic codes
of length $n$ over $\mathbb{F}_q$.
Theoretically,  the irreducible factors of $x^n-1$ in $\mathbb{F}_q[x]$ can be derived  by the $q$-cyclotomic cosets modulo $n$.
For any integer $t$, the {\em $q$-cyclotomic coset $C_t$ of $t$ modulo $n$} is
defined  by
$$
C_t=\Big\{tq^j~(\bmod~{n})\,\Big{|}\,t=0,1,\cdots\Big\}.
$$
Take $\alpha$ (maybe in some extension field of $\mathbb{F}_q$)  to be a primitive $n$-th root of unity, which means that
$n$ is the smallest positive integer such that $\alpha^n=1$.

Let $C_0=\{0\}, C_{i_1}, C_{i_2}, \cdots, C_{i_t}$ be all the distinct $q$-cyclotomic
cosets modulo $n$. It is well known that
$$
x^n-1=\big(x-1\big)M_1(x)M_2(x)\cdots M_{t}(x)
$$
with
$$
M_{j}(x)=\prod\limits_{s\in C_{i_j}}\big(x-\alpha^{s}\big),~~~1\leq j\leq t,
$$
all being monic irreducible in $\mathbb{F}_q[x]$.
The {\em defining set} of   $\mathcal{C}=\langle g(x)\rangle$
is the subset of integers
$$
Z=\Big\{j\,\Big{|}\,0\leq j\leq n-1, g(\alpha^j)=0\Big\}.
$$
It is readily seen that the defining set $Z$ is a union of
$q$-cyclotomic cosets modulo $n$.

A linear code of length $n$ over $\mathbb{F}_q$ is called an $[n,k,d]$
code if its dimension
is $k$ and minimum (Hamming) distance is $d$.
The following results are  well known.
\begin{lem}(BCH bound for cyclic codes)\label{BCH}
Let $\mathcal{C}$ be a cyclic code of length $n$ over $\mathbb{F}_{q}$.
Let $\alpha$ be  a primitive $n$-th  root of unity in some extension field of $\mathbb{F}_{q}$.
Assume the generator polynomial of $\mathcal{C}$ has roots that include the set $\{\alpha^{i}\,|\, i_1\leq i\leq i_1+d-2\}$.
Then the minimum distance of $\mathcal{C}$ is at least $d$.
\end{lem}

\begin{Proposition}\label{Singleton}
(Singleton  bound)   If $\mathcal{C}$ is  an $[n,k,d]$ linear code over $\mathbb{F}_{q}$, then
$d\leq n-k+1$.
\end{Proposition}

A linear code achieving this   Singleton bound is called a  {\em maximum distance separable
(MDS) code}.
A remark is in order at this point.
Lemma \ref{BCH} and Proposition \ref{Singleton}  provide  a useful method to construct MDS cyclic codes:
If the generator polynomial of a cyclic code  $\mathcal{C}$   has roots precisely equal to  the set $\{\alpha^{i}\,|\, i_1\leq i\leq i_1+d-2\}$,
then the minimum distance   of $\mathcal{C}$ is exactly equal to  $d$.
In particular, $\mathcal{C}$ is an MDS cyclic code with parameters $[n,n-d+1,d]$.
Indeed, it follows from Lemma \ref{BCH} that the minimum distance of $\mathcal{C}$ is at least $d$. Since
the dimension of $\mathcal{C}$ is equal to $n-d+1$, then the minimum distance of $\mathcal{C}$ is no more than
$n-(n-d+1)+1=d$, which implies that $\mathcal{C}$ is an MDS cyclic code with parameters $[n,n-d+1,d]$ (e.g., see \cite{Chen}).
We will construct    MDS cyclic codes based on this fact.

\section{A Necessary and Sufficient Condition}

Throughout  this paper, $\mathcal{C}_0$ denotes the cyclic code of length $n$ over $\mathbb{F}_q$
with generator polynomial $1+x+\cdots+x^{n-1}$, i.e.,
$$
\mathcal{C}_0=\Big\{\alpha\big(1+x+\cdots+x^{n-1}\big)\,\Big{|}\,\alpha\in \mathbb{F}_q\Big\}.
$$
We always assume that $\mathcal{C}$ is a
cyclic code of length $n$ over $\mathbb{F}_q$ with parity-check polynomial $h(x)$ such that $h(1)=0$.
Note that $\mathcal{C}_0$ is contained in $\mathcal{C}$ if and only if $h(1)=0$.

Two codewords $c_1(x),c_2(x)$ of $\mathcal{C}$ are said to be {\em equivalent} if there exists an integer $t$
such that $x^tc_1(x)\equiv c_2(x)\pmod{x^n-1}$. The codewords of $\mathcal{C}$ then are classified into equivalence classes.
We say that a codeword $c(x)\in \mathcal{C}$ has size $n$ if  the equivalence class   containing  $c(x)$  has size $n$.

The following result establishes
a necessary and sufficient condition
under which  any codeword  in $\mathcal{C}\setminus\mathcal{C}_0$ has size $n$.

\begin{Theorem}\label{theorem1}
Let $\mathcal{C}$ be    a
cyclic code of length $n$ over $\mathbb{F}_q$  with parity-check polynomial $h(x)$ such that $h(1)=0$, say $h(x)=(x-1)h'(x)$.
Then the following statements are equivalent:
\begin{itemize}
\item[(i)]
Any  codeword  of
$
\mathcal{C}\setminus\mathcal{C}_0
$
has size $n$.
\item[(ii)] Any   root of $h'(x)$ is a primitive $n$-th root of unity.
\end{itemize}
\end{Theorem}
\begin{proof}
Assume first that $(i)$ holds.
If $n$ is a prime number, then all the   roots of $(x^n-1)/(x-1)$ are primitive $n$-th roots of unity. In particular,
the   roots of $h'(x)$ are primitive $n$-th roots of unity, and we are done. Therefore, we can assume that
$n$ is a composite  number.
To get $(ii)$, it is enough to prove that the roots of $(x^n-1)/(x-1)$  which have  multiplicative order less than $n$ are roots of $g(x)=(x^n-1)/h(x)$.

Let $\alpha$ be a primitive $n$-th root of unity. Assume to the contrary that there exists $0<i_0<n$ with $\gcd(i_0,n)\neq1$ such that
$\alpha^{i_0}$ is not a root of $g(x)$,
i.e.,  $\alpha^{i_0}$ is a root of $(x^n-1)/(x-1)$ having order (say $r$)  less than $n$ and $g(\alpha^{i_0})\neq0$.
 Let $C_{i_0}$ denote the $q$-cyclotomic coset modulo $n$ containing $i_0$, and
assume that  $C_0=\{0\}, C_{i_0}, C_{i_1}, \cdots, C_{i_t}, C_{i_{t+1}}, \cdots, C_{i_{t+s}}$ are all the distinct $q$-cyclotomic
cosets modulo $n$.  Without loss of generality, suppose that  the defining set of $\mathcal{C}$ is
 $C_{i_1}\bigcup C_{i_2}\bigcup\cdots\bigcup C_{i_t}$, i.e.,
$$
g(x)=\prod\limits_{j=1}^t M_{j}(x),~~~ \hbox{where~$M_{j}(x)=\prod\limits_{k\in C_{i_{j}}}(x-\alpha^k)$~for~$1\leq j\leq t$}.
$$
Now let
$$
m(x)=\prod\limits_{j=1}^s M_{t+j}(x),~~~ \hbox{where~$M_{t+j}(x)=\prod\limits_{k\in C_{i_{t+j}}}(x-\alpha^k)$~for~$1\leq j\leq s$}.
$$
It follows that $c(x)=m(x)g(x)$ lies in  $\mathcal{C}\setminus \mathcal{C}_0$.
However, from our construction,
$$
(x^r-1)m(x)g(x)\equiv0\pmod{x^n-1}.
$$
This says that $c(x)$ is a  codeword of $\mathcal{C}\setminus \mathcal{C}_0$, and the size of the equivalence class containing  $c(x)$
is less than $n$. This is a contradiction.

Assume that $(ii)$ holds.
Suppose otherwise that there exists a   codeword $c_1(x)\in \mathcal{C}\setminus \mathcal{C}_0$
such that the equivalence class containing it has size $t<n$.
Thus
$$
x^tc_1(x)\equiv c_1(x)\pmod{x^n-1}.
$$
Let $c_1(x)=m_1(x)g(x)$, then we  have
\begin{equation*}
(x^t-1)m_1(x)g(x)\equiv 0\pmod{x^n-1},
\end{equation*}
and hence
\begin{equation}\label{equ1}
(x^t-1)m_1(x)\equiv 0\pmod{h'(x)}.
\end{equation}
Recall that all the   roots of $h'(x)$ are primitive $n$-th roots of unity.
Since the roots of $x^t-1$ contain no primitive $n$-th root of unity,   $h'(x)$ is a divisor of $m_1(x)$.
We have arrived at a contradiction since this implies $c_1(x)=m_1(x)g(x)\in \mathcal{C}_0$.
\end{proof}

\begin{Example}\rm
Consider cyclic codes of length $9$ over $\mathbb{F}_8$. It is easy to verify that all the distinct $8$-cyclotomic cosets
modulo $9$ are given by
$C_0=\{0\}, C_1=\{1,8\}, C_2=\{2,7\}, C_3=\{3,6\}$ and $C_4=\{4,5\}$.
Take $\alpha$ to be a primitive ninth root of unity in $\mathbb{F}_{64}$. Then
$$
x^9-1=(x-1)M_1(x)M_2(x)M_3(x)M_4(x), ~~\hbox{with}~~M_i(x)=\prod\limits_{j\in C_i}(x-\alpha^j),~1\leq i\leq4
$$
gives the irreducible factorization of $x^9-1$ over $\mathbb{F}_8$. It is readily seen that the roots of $M_3(x)$ are primitive
third roots of unity. Let $\mathcal{C}_1$ and $\mathcal{C}_2$ be   cyclic codes of length $9$ over $\mathbb{F}_8$ with parity-check polynomials $(x-1)M_1(x)M_2(x)M_4(x)$
and $(x-1)M_1(x)M_2(x)$, respectively.
It follows from Theorem \ref{theorem1} that any codeword of  $\mathcal{C}_i\setminus \mathcal{C}_0$
has size $9$, for $i=1,2$. Note that $\mathcal{C}_2$ is an MDS cyclic code with parameters $[9,5,5]$.
\end{Example}

As we will show in the next section,
Theorem \ref{theorem1} is useful in constructing optimal   FHS sets.
In the rest of this section,
Theorem \ref{theorem1} is used to  answer an open question posed by Ding {\it et al.} \cite[p.3302]{Ding09}.

Let $m$ be a positive integer such that $\gcd(m,q-1)=1$.
Take $n=(q^m-1)/(q-1)$. Let  $\beta$ be a generator of $\mathbb{F}_{q^m}^*=\mathbb{F}_{q^m}\setminus\{0\}$ and $\gamma=\beta^{q-1}$.
Define a cyclic   code $\mathcal{C}_{(q,m)}$ by
$$
\mathcal{C}_{(q,m)}=\Big\{c(x)\,\Big{|}\,c(x)\in \mathbb{F}_q[x]_n~\hbox{and}~c(\gamma)=0\Big\}
$$
where $\mathbb{F}_q[x]_n$ consists of all polynomials of degree at most $n-1$ over $\mathbb{F}_q$.
It is  known that $\mathcal{C}_{(q,m)}$ is an $[n,n-m,3]$ code with defining set
\begin{equation}\label{defining}
C_1=\Big\{q^i \pmod{n} \,\Big{|}\,0\leq i\leq m-1\Big\}.
\end{equation}
Assuming that $n=(q^m-1)/(q-1)$ is a prime number, Ding {\it et al.} in \cite[Theorem 12]{Ding09} used $\mathcal{C}_{(q,m)}$ to construct an  FHS set whose paramters meet the sphere-packing bound   (\ref{sbound}).
The following open question was posed in \cite[p.3302]{Ding09}: ``If the condition that $n=(q^m-1)/(q-1)$ is a prime number is dropped, is it
still true that any    codeword of $\mathcal{C}_{(q,m)}\setminus\mathcal{C}_0$ has size $n$?"

By Theorem \ref{theorem1}, we   have a complete answer  to this question.

\begin{Corollary}\label{open}
Assume the same notation as previously defined. Then
 any codeword of $\mathcal{C}_{(q,m)}\setminus\mathcal{C}_0$  has size $n$ if and only if $n$ is a prime number.
\end{Corollary}
\begin{proof}
It is known that $\mathcal{C}_{(q,m)}$ is a cyclic code of length $n$ over $\mathbb{F}_q$ with defining set given by (\ref{defining}).
Thus the generator polynomial of   $\mathcal{C}_{(q,m)}$ is
$$
g(x)=\prod\limits_{i\in C_1}(x-\gamma^i).
$$
Let $h(x)=\frac{x^n-1}{g(x)}=(x-1)h'(x)$.
Suppose that    any codeword  of $\mathcal{C}_{(q,m)}\setminus\mathcal{C}_0$ has size $n$.
It follows from Theorem \ref{theorem1} that all the roots of $h'(x)$ are primitive $n$-th roots of unity.
Suppose otherwise that $1<d<n$ is a divisor of $n$.
Then $\gamma^d$ must be a root of $h'(x)$,   contradicting Theorem \ref{theorem1}.
\end{proof}

The following result can be proven in a fashion similar to Theorem \ref{theorem1}.
\begin{Proposition}\label{pro}
Let $\mathcal{D}$ be a nonzero cyclic code of length $n$ over $\mathbb{F}_q$ with parity-check polynomial $h(x)$,
where $n$ is a positive integer co-prime to $q$.
Then any  nonzero codeword of
$
\mathcal{D}
$
has size $n$ if and only if the roots of $h(x)$ are   primitive $n$-th roots of unity.
\end{Proposition}

Note that the result of Proposition \ref{pro} appeared previously in \cite[Theorem 1]{Le}.
We will apply Theorem \ref{theorem1} and Proposition \ref{pro} to construct optimal  FHS sets.

\section{Optimal   FHS Sets from Cyclic Codes}
In  this section, three families of optimal   FHS sets with respect to the Singleton  bound  (\ref{sinbound})  are obtained.
It turns out that some of the   FHS sets are also optimal with respect to the Peng-Fan bound     (\ref{PF}) at the same time.
In  the  light of \cite{Ding09} and \cite{Ding10}, we first present two  basic facts.
\begin{itemize}
\item[Fact 1.]

If one has found an $[n,k,n-k+1]$ MDS cyclic code $\mathcal{C}$ over $\mathbb{F}_q$ with $n>q$
satisfying Theorem \ref{theorem1}, then   an  FHS set is obtained by taking exactly one element from  every equivalence class of $\mathcal{C}\setminus\mathcal{C}_0$,
which results  in an $(n,\frac{q^{k}-q}{n}, k-1;q)$ FHS set. It is straightforward to verify that this   FHS set is optimal with respect to the
 Singleton bound (\ref{sinbound}).

\item[Fact 2.]

Similarly, if one has found an $[n,k,n-k+1]$ MDS cyclic code $\mathcal{D}$ over $\mathbb{F}_q$
satisfying Proposition \ref{pro}, then an FHS set is obtained by taking exactly one element  from every equivalence class of $\mathcal{D}\setminus\{0\}$,
which results in an $(n,\frac{q^{k}-1}{n}, k-1;q)$ FHS set. This FHS set is  optimal with respect to the
 Singleton bound (\ref{sinbound}).
\end{itemize}

We  explain the reason for which an $[n,k,n-k+1]$ MDS cyclic code over $\mathbb{F}_q$ with $n>q$
satisfying Theorem \ref{theorem1} gives rise to   an  $(n,N,\lambda;\ell)=(n,\frac{q^{k}-q}{n}, k-1;q)$ FHS set, i.e., Fact 1 holds.
It follows from the  definition  of $\lambda$ and the Singleton bound
(\ref{sinbound}) that  $\lambda\leq n-(n-k+1)=k-1$
and   $N=\frac{q^{k}-q}{n}\leq\lfloor\frac{q^{\lambda+1}}{n}\rfloor$, which forces $\lambda=k-1$.  Fact 2 is also satisfied
    for a similar reason.

In the following, three families of MDS cyclic codes satisfying  Fact $1$ or Fact $2$ are constructed
and their parameters  are computed. Consequently,
optimal   FHS sets with respect to the Singleton  bound    (\ref{sinbound})  are derived from these MDS cyclic codes.
These   FHS sets are new in the sense that their parameters have not been covered
in the literature.

\subsection{Optimal FHS sets of length $q+1$}

Consider cyclic codes of length $q+1$ over $\mathbb{F}_q$.
We first study the case  $q=2^m$, where $m>1$ is a  positive integer.
Let $s=2^{m-1}$. By \cite[p.324]{Mabook}, we know that   all the distinct $q$-cyclotomic cosets modulo $q+1$ are given by
$$
C_0=\big\{0\big\}~~\hbox{and}~~C_i=\big\{i, q+1-i\big\},~~~\hbox{for~$1\leq i\leq s$}.
$$
Suppose $p$ is the smallest prime divisor of $q+1$.
Let $\alpha\in \mathbb{F}_{q^2}$ be a primitive $(q+1)$-st root of unity.
For any integer $k$ with $1\leq k\leq\min\{p-1,s\}$,
let $\mathcal{C}$ be a cyclic code of length $q+1$ over $\mathbb{F}_q$ with
parity-check polynomial
$$
h(x)=(x-1)M_1(x)M_2(x)\cdots M_{k}(x), ~~~M_j(x)=\prod\limits_{t\in C_j}(x-\alpha^t), ~~1\leq j\leq k.
$$
Equivalently, $\mathcal{C}$ is the cyclic code of length $q+1$ over $\mathbb{F}_q$ with defining set
$$
C_{k+1}\bigcup C_{k+2}\bigcup\cdots\bigcup C_{s}.
$$
Thus $\mathcal{C}$ is a $[q+1, 2k+1,q-2k+1]$ MDS cyclic code satisfying Theorem \ref{theorem1}.
The above discussion leads to the following result.
\begin{Theorem}\label{construction1}
Let $q=2^m$, where $m>1$ is a positive integer. Suppose $p$ is the smallest prime divisor of $q+1$.
For any integer $k$ with $1\leq k\leq\min\{p-1,2^{m-1}\}$, we have a  $(q+1,\frac{q^{2k+1}-q}{q+1},2k;q)$  FHS set whose parameters
meet the Singleton bound   (\ref{sinbound}).
\end{Theorem}

As an additional remark, we observe  that if $m$ is odd in Theorem \ref{construction1}, then $p=3$. This is simply because $3$ is always a divisor
of $2^m+1$ for any odd $m$.

\begin{Example}\rm
Take $q=2^3=8$ and $n=q+1=9$. Then $p=3$ and $1\leq k\leq 2$.
Thus we have a $(9,(8^{2k+1}-8)/9,2k;8)$ FHS set
for $k=1,2$. It is easy to verify that these parameters indeed meet the Singleton bound   (\ref{sinbound}).
\end{Example}

\begin{Example}\rm
Take $q=2^4=16$ and $n=q+1=17$. We have $p=17$ and $\min\{p-1,2^3\}=8$, thus $1\leq k\leq 8$.
Theorem \ref{construction1} applies to give an optimal $(17,\frac{16^{2k+1}-16}{17},2k;16)$  FHS set
for any $1\leq k\leq 8$ with respect to the Singleton bound
(\ref{sinbound}).
\end{Example}

We next consider cyclic codes of length $q+1$ over $\mathbb{F}_q$ in the case where $q$ is an odd prime power.
It is easy to verify that   all the distinct $q$-cyclotomic cosets modulo $q+1$ are given by
$$
C_0=\big\{0\big\},~~C_{\frac{q+1}{2}}=\big\{\frac{q+1}{2}\big\}~~\hbox{and}~~C_i=\big\{i, q+1-i\big\},~~~\hbox{for~$1\leq i\leq \frac{q+1}{2}-1$}.
$$
Let $\alpha\in \mathbb{F}_{q^2}$ be a primitive $(q+1)$-th root of unity.
In this case, let $\mathcal{C}$ be a cyclic code of length $q+1$ over $\mathbb{F}_q$ with
parity-check polynomial $h(x)=(x-1)M_1(x)$, where $M_1(x)=(x-\alpha)(x-\alpha^{-1})$.
Thus $\mathcal{C}$ is a $[q+1, 3,q-1]$ MDS cyclic code satisfying Theorem \ref{theorem1}.
From Fact 1,
there is an optimal $(q+1, q(q-1),2;q)$  FHS set
with respect to  the Singleton bound (\ref{sinbound}), where $q$ is an odd prime power.
We claim that these parameters also meet the Peng-Fan bound      (\ref{PF}).  To see this, note that
$M(\mathcal{F})=2$ and $I=\lfloor nN/\ell\rfloor=q^2-1$.
It is clear that
$$
\left\lceil\frac{2InN-(I+1)I\ell}{(nN-1)N}\right\rceil=\left\lceil\frac{2q(q+1)^2(q-1)^2-q^3(q+1)(q-1)}{q(q-1)\big(q(q+1)(q-1)-1\big)}\right\rceil.
$$
After simple computations we have
$$
\frac{2q(q+1)^2(q-1)^2-q^3(q+1)(q-1)}{q(q-1)\big(q(q+1)(q-1)-1\big)}>1,
$$
proving the claim.

Summarizing the previous discussion, we arrive at the following result.

\begin{Theorem}\label{construction1'}
Let $q$ be an odd prime power.
Then there is a $(q+1, q(q-1),2;q)$ FHS set whose parameters
meet both the Peng-Fan bound   (\ref{PF}) and the Singleton bound   (\ref{sinbound}).
\end{Theorem}

\begin{Example}\rm
Take $q=5^2=25$ and $n=25+1=26$.
Theorem \ref{construction1'} applies to give a $(26,600,2;25)$ FHS set. It is easy to check that both
the Peng-Fan bound   (\ref{PF}) and the Singleton bound    (\ref{sinbound})  are met.
\end{Example}

\subsection{Optimal FHS sets of length dividing $q+1$}
Let $n>1$    be an odd   divisor of $q+1$.
Let $\alpha\in \mathbb{F}_{q^2}$ be a primitive $n$-th root of unity.
It is easy to verify that   all the distinct $q$-cyclotomic cosets modulo $n$ are given by
$$
C_0=\big\{0\big\}~~\hbox{and}~~C_i=\{i, n-i\},~~~\hbox{for~$1\leq i\leq \frac{n-1}{2}$}.
$$
Let   $M\leq(n-1)/2$   be the largest positive  integer satisfying $\gcd(M,n)\neq1$; if no such integer    exists,
then $M$ is assumed to be $0$. Observe that $\gcd(\frac{n-1}{2},n)=1$, thus $M\leq(n-3)/2$. Fix a value $k$, $0\leq k\leq(n-3)/2-M$.
Let $\mathcal{D}$ be a cyclic code of length $n$ over $\mathbb{F}_q$ with
parity-check polynomial
$$
h(x)=M_{(n-1)/2}(x)M_{(n-1)/2-1}(x)\cdots M_{(n-1)/2-k}(x)
$$
where $M_j(x)=\prod_{t\in C_j}(x-\alpha^t), ~~(n-1)/2-k\leq j\leq (n-1)/2$.
It follows that  $\mathcal{D}$ is an $[n, 2(k+1),n-2k-1]$ MDS cyclic code satisfying Proposition \ref{pro}.

The discussion above establishes the following theorem.

\begin{Theorem}\label{theoremlasttwo}
Let    $n>1$ be an odd divisor of $q+1$.
Let $M\le (n-3)/2$ be the largest integer such that $\gcd(M,n)>1$.
For any integer  $k$  with $0\leq k\leq(n-3)/2-M$, we have an optimal $\big(n,(q^{2k+2}-1)/n,2k+1;q\big)$  FHS set
with respect to the Singleton bound   (\ref{sinbound}).

In particular, by taking $k=0$,
we have an $\big(n,(q^{2}-1)/n,1;q\big)$ FHS set, whose parameters
meet both the Peng-Fan bound   (\ref{PF}) and the Singleton bound    (\ref{sinbound}).
\end{Theorem}

Note that the parameters  $\big(n,(q^{2}-1)/n,1;q\big)$  appeared previously in \cite[Corollary 12]{Ding10}.
In comparison, our
method is quite neat and clearer for understanding.

\begin{Remark}\rm
Theorems \ref{construction1}, \ref{construction1'} and \ref{theoremlasttwo}
generate three new families of optimal FHS sets  with respect to the Singleton bound  (\ref{sinbound}),
some of which  are also optimal with respect to the Peng-Fan bound (\ref{PF1}).
Most  previously known optimal FHS sets in the literature are constructed with  respect to the Peng-Fan bound (e.g., see \cite[TABLE I]{Yang}
or \cite[TABLE II]{Zeng12}). Optimal FHS sets with respect to the Singleton bound have been mainly studied in \cite{Ding09}, \cite{Ding10} and \cite{Yang}.
\cite[Corollary 12]{Ding10} gives the first class of FHS sets  simultaneously meeting  the Peng-Fan bound and the Singleton bound.
Theorem \ref{construction1'} gives one more example of such FHS sets.
\end{Remark}

\begin{Example}\rm
Take $q=2^9=512$, then $q+1=513=3^3\times19$. Choose $n=27$, thus $M=12$ and $k=0$.
It follows from Theorem \ref{theoremlasttwo} that  we have a $(27,9709,1;512)$ FHS set, whose parameters
meet both the Peng-Fan bound   (\ref{PF}) and the Singleton bound   (\ref{sinbound}).
\end{Example}

\begin{Example}\rm
Take $q=2^5=32$, then $q+1=33=3\times11$. Choose $n=11$, thus $M=0$ and $0\leq k\leq 4$.
It follows from Theorem \ref{theoremlasttwo} that  we have an $(11,(32^{2k+2}-1)/11,2k+1;32)$ FHS set
for any $0\leq k\leq 4$, whose parameters meet the Singleton bound (\ref{sinbound}).
\end{Example}

{\bf Appendix}

We give a detailed proof of Theorem \ref{p-f-iden}.

{\bf Proof of Theorem \ref{p-f-iden}:}
Let
$$
PF1=\frac{(nN-\ell)n}{(nN-1)\ell}~~~\hbox{and}~~~PF2=\frac{2InN-(I+1)I\ell}{(nN-1)N}.
$$
We want to prove that $\left\lceil PF1\right\rceil=\left\lceil PF2\right\rceil$.
If $nN=\ell$, then $I=1$, thus $\left\lceil PF1\right\rceil=0$ and $\left\lceil PF2\right\rceil=0$.
Hereafter, we assume that $nN>\ell$.

Suppose $nN=\ell I+J$ with $0\leq J<\ell$.
\cite[Proposition 13]{Yang} says that $\left\lceil PF2\right\rceil\geq\left\lceil PF1\right\rceil$,
\begin{equation}\label{yang}
\begin{split}
PF2-PF1&=\frac{2InN-(I+1)I\ell}{(nN-1)N}-\frac{(nN-\ell)n}{(nN-1)\ell}\\
&=\frac{2InN\ell-(I+1)I\ell^2}{(nN-1)N\ell}-\frac{(nN-\ell)nN}{(nN-1)\ell N}\\
&=\frac{2(nN-J)nN-(nN-J+\ell)(nN-J)-(nN-\ell)nN}{(nN-1)\ell N}\\
&=\frac{(\ell-J)J}{(nN-1)\ell N}\\
&\geq0.
\end{split}
\end{equation}
If $J=0$, i.e., $\ell$ is a divisor of $nN$, then $PF1=PF2$, and we are done. We can assume, therefore, that
$\ell$ does not divide  $nN$. In particular, $\ell>1$.
Note that $I=\lfloor nN/\ell\rfloor\geq1$. Since $nN=I\ell+J$ with $J\neq0$,
we have $nN>I\ell$. Clearly, $2I\geq I+1$, which implies that  $2InN-(I+1)I\ell>0$ and  $\left\lceil PF2\right\rceil\geq1$.

We first consider the case    $n\leq\ell$.
In this case, $\lceil PF1\rceil=1$.
Our task is thus to show that $\lceil PF2\rceil=1$. To this end, it is enough to show that
\begin{equation}\label{quadratic}
nN^2-\big(2In+1\big)N+I^2\ell+I\ell>0.
\end{equation}
This is obvious because,   regarding the expression on the left hand side of (\ref{quadratic})  as a quadratic in the variable $N$,
its discriminant $\Delta$ is less than $0$,  i.e.,
\begin{equation*}
\Delta=4\big(I^2n^2-I^2\ell n\big)+4\big(In-In\ell\big)+1
<0.
\end{equation*}
Therefore, we have obtained that
$\lceil PF1\rceil=\lceil PF2\rceil=1$ in the case where   $nN>\ell$ and $n\leq\ell$.

Now let $n=s\ell+r$ with $s>0$ and  $0<r\leq\ell-1$ (recall that $\ell$ is not a divisor of $n$).
Let $rN=t\ell+J$ with $t\geq0$ and  $0<J\leq \ell-1$.
We then see that
\begin{equation}\label{pf1}
PF1=\frac{(nN-1+1-\ell)n}{(nN-1)\ell}=s+\frac{r}{\ell}-\frac{(\ell-1)n}{(nN-1)\ell}.
\end{equation}

Combining  Equations (\ref{yang}) and (\ref{pf1}),
\begin{equation}\label{pf2}
PF2=PF1+\frac{(\ell-J)J}{(nN-1)\ell N}
=s+\frac{r}{\ell}+\frac{(\ell-J)J}{(nN-1)\ell N}-\frac{(\ell-1)n}{(nN-1)\ell}.
\end{equation}

We now consider two subcases separately.

Subcase 1: $\frac{r}{\ell}\leq\frac{(\ell-1)n}{(nN-1)\ell}$, i.e., $r(nN-1)\leq(\ell-1)n.$
In this subcase,
$$
rnN-r=r(nN-1)\leq(\ell-1)n<\ell n-r,
$$
which implies $rN<\ell.$ This means that $t=0$ and $J=rN$. Clearly,
$$
\frac{(\ell-1)n}{(nN-1)\ell}-\frac{r}{\ell}=\frac{(\ell-1)n-r(nN-1)}{(nN-1)\ell}.
$$
It is readily seen that  $(\ell-1)n-r(nN-1)<(nN-1)\ell.$
Thus
$$
-1<\frac{r}{\ell}-\frac{(\ell-1)n}{(nN-1)\ell}\leq0.
$$
This leads to
$
\lceil PF1\rceil=s.
$
In order to show that $\lceil PF2\rceil=\lceil PF1\rceil=s$, by Equation (\ref{pf2}), it suffices to prove that
$$
\frac{r}{\ell}+\frac{(\ell-J)J}{(nN-1)\ell N}-\frac{(\ell-1)n}{(nN-1)\ell}=\frac{rN(nN-1)+(\ell-rN)rN-(\ell-1)nN}{(nN-1)\ell N}\leq0,
$$
or equivalently,
$$
r(nN-1)+(\ell-rN)r-(\ell-1)n
=(n-r)(rN-\ell+1)\\
\leq0.
$$
We are done  because $rN<\ell$ and $n>r$.

Subcase 2: $\frac{r}{\ell}>\frac{(\ell-1)n}{(nN-1)\ell}$.
In this subcase, it is clear that
$
\lceil PF1\rceil=s+1.
$
We claim that
$$
\frac{r}{\ell}+\frac{(\ell-J)J}{(nN-1)\ell N}-\frac{\ell-1}{(nN-1)\ell}\leq1,
$$
then
$
\left\lceil PF2\right\rceil=s+1=\left\lceil PF1\right\rceil
$
and this will complete the proof.

Since
$$
\frac{(\ell-J)J}{(nN-1)\ell N}=
\frac{(\ell+t\ell-rN)(rN-t\ell)}{(nN-1)\ell N}
=\frac{rN\ell+2rNt\ell-r^2N^2-t\ell^2-t^2\ell^2}{(nN-1)\ell N},
$$
it follows that
$$
\frac{r}{\ell}+\frac{(\ell-J)J}{(nN-1)\ell N}-\frac{(\ell-1)n}{(nN-1)\ell}=\frac{rnN^2-rN+rN\ell+2rNt\ell-r^2N^2-t\ell^2-t^2\ell^2-\ell nN+nN}{(nN-1)\ell N}.
$$
We are left to prove that
\begin{equation}\label{equation}
\ell nN^2\geq rnN^2-rN+rN\ell+2rNt\ell-r^2N^2-t\ell^2-t^2\ell^2-\ell nN+nN+\ell N.
\end{equation}
We prove  Inequality (\ref{equation}) by induction on $t\geq0$.
For the base step $t=0$,
\begin{equation*}
\begin{split}
&\ell nN^2-\big(rnN^2-rN+rN\ell-r^2N^2-\ell nN+nN+\ell N\big)\\
&=n(\ell N^2-rN^2)+rN-rN\ell+r^2N^2+N(\ell n-n-\ell)\\
&\geq(\ell+1)(\ell N^2-rN^2)-rN\ell+r^2N^2
\end{split}
\end{equation*}
where the last inequality holds because $n>\ell>r$.
To obtain the desired result, it is enough to show that
$
(\ell+1)(\ell N-rN)-r\ell+r^2N\geq0.
$
By  a simple computation:
$$
(\ell+1)(\ell N-rN)-r\ell+r^2N
\geq\ell^2-2\ell r+\ell-r+r^2
=\big(\ell-r\big)^2+\ell-r
>0.
$$

Now suppose  Inequality (\ref{equation}) holds true for any $t>0$.
For the
inductive step,
\begin{equation*}
rnN^2-rN+rN\ell+2(t+1)rN\ell-r^2N^2-(t+1)\ell^2-(t+1)^2\ell^2-\ell nN+nN+\ell N
\end{equation*}
is equal to
$$
rnN^2-rN+rN\ell+2trN\ell-r^2N^2-t\ell^2-t^2\ell^2-\ell nN+nN+\ell N+(2rN\ell-2\ell^2-2t\ell^2).
$$
Note that
$$
2rN\ell-2\ell^2-2t\ell^2=2\ell\big(rN-\ell t-\ell\big)=2\ell(J-\ell)<0.
$$
By the
inductive hypothesis,
$$
rnN^2-rN+rN\ell+2(t+1)rN\ell-r^2N^2-(t+1)\ell^2-(t+1)^2\ell^2-\ell nN+nN+\ell N\leq\ell nN^2.
$$
This completes the proof.

\noindent{\bf Acknowledgements}
The authors are  grateful to   the three anonymous reviewers for their  helpful  comments and
suggestions that improved an earlier version of this paper. The research of the first and third authors
was partially supported by NTU Research Grant M4080456. The fourth author was supported by NSFC (Grant No. 11171370) and self-determined research funds of CCNU from the colleges's basic research and operation of MOE (Grant No.~CCNU14F01004).
Part of this work was done when B. Chen was with Nanyang Technological University.

\end{document}